\documentclass[11pt,letterpaper]{article}

\usepackage[letterpaper, left=1in, right=1in, top=0.9in, bottom=0.9in]{geometry}
\usepackage[american]{babel}
\usepackage[normalem]{ulem}
\usepackage{amsmath, amssymb, cases, amsthm}
\usepackage{thmtools}
\usepackage[shortlabels]{enumitem}
\usepackage{mdframed}
\usepackage{bbm}
\usepackage{bm}
\usepackage{microtype}
\usepackage{xcolor}
\usepackage{makecell}
\usepackage{mathtools}
\usepackage{algorithmic}
\usepackage[procnumbered,ruled,vlined,linesnumbered]{algorithm2e}
\usepackage{float}
\usepackage{varwidth}
\usepackage{modletter}
\usepackage{tcolorbox}
\newtcolorbox{construction}[2][]
{
	colframe = gray!50,
	colback  = gray!10,
	coltitle = gray!10!black,
	left*=0mm, 
	before skip = 10pt,
	after skip = 10pt,
	title    = \textbf{\space\space #2},
	#1,
}

\SetKwInput{KwData}{Input}
\SetKwInput{KwResult}{Output}
\SetKwInput{KwGlobalVar}{Global variables}

\usepackage[bookmarks,colorlinks,breaklinks]{hyperref}
\hypersetup{urlcolor=blue, colorlinks=true, citecolor=green!50!black, linkcolor=blue}
\usepackage[capitalize,nosort,nameinlink]{cleveref}

\declaretheorem[numberwithin=section,refname={Theorem,Theorems},Refname={Theorem,Theorems}]{theorem}

\declaretheorem[numberlike=theorem]{lemma}

\declaretheorem[numberlike=theorem]{corollary}
\declaretheorem[numberlike=theorem]{definition}

\theoremstyle{definition}

\def\final{0}  %
\ifnum\final=0  %
\newcommand{\todo}[1]{{\color{red}[{\tiny TODO: \bf #1}]\marginpar{\color{red}*}}}
\newcommand{\yonggang}[1]{{\color{blue}[{\tiny Yonggang: \bf #1}]\marginpar{\color{blue}*}}}
\else %
\newcommand{\yonggang}[1]{}
\newcommand{\todo}[1]{}
\fi

\newcommand{\eps}{\epsilon}

\newcommand{\polylog}{\mathrm{polylog}}
\newcommand{\dist}{\mathrm{dist}}

\begin{document}
\sloppy

\title{
Reducing Shortcut and Hopset Constructions to Shallow Graphs
}
\author{
Bernhard Haeupler\thanks{INSAIT, Sofia University ``St. Kliment Ohridski'' \& ETH Zürich, \texttt{haeupler.b@gmail.com}. Partially funded by the Ministry of Education and Science of Bulgaria's support for INSAIT as part of the Bulgarian National Roadmap for Research Infrastructure and through the European Research Council (ERC) under the European Union's Horizon 2020 research and innovation program (ERC grant agreement 949272).} \and 
Yonggang Jiang\thanks{MPI-INF \& Saarland University,  \texttt{yjiang@mpi-inf.mpg.de}.} \and 
Thatchaphol Saranurak\thanks{University of Michigan, \texttt{thsa@umich.edu}. Supported by NSF Grant CCF-2238138.}
}
\date{}
\maketitle

\pagenumbering{gobble}
\begin{abstract}

We introduce a blackbox framework that simplifies all known parallel algorithms with near-linear work for single-source reachability and shortest paths in directed graphs. 
Specifically, existing reachability algorithms rely on constructing \emph{shortcuts}; our blackbox allows these algorithms that construct shortcuts with hopbound $h$ to assume the input graph $G$ is ``shallow'', meaning if vertex $s$ can reach vertex $t$, it can do so in approximately $h$ hops. This assumption significantly simplifies shortcut construction \cite{fineman2018nearly,LiuJS19}, resulting in simpler parallel reachability algorithms. Furthermore, our blackbox extends naturally to simplify parallel algorithms for constructing \emph{hopsets} and, consequently, for computing shortest paths \cite{cao2020efficient,cao2023parallel,rozhovn2023parallel}.

\end{abstract}

\clearpage

\pagenumbering{arabic}

\newcommand{\cAD}{\cA_{\mathrm{DAG}}}

\newcommand{\rd}{reachability diameter}

\section{Introduction}

Single-source reachability and shortest paths are among the most fundamental graph problems. Their (near-linear) time algorithms, such as breadth-first search, depth-first search, and Dijkstra, have been the cornerstone of graph algorithms. Unfortunately, all these classical algorithms are inherently sequential. 

This inspired an extensive line of work on parallel algorithms for these problems since the 80s \cite{UllmanY91,Klein93,KleinS97,Spencer97}. These parallel algorithms are more parallel and achieve sublinear depth, but they still require work that is polynomially larger than linear. The major open problem in this area is whether there exists an algorithm with \emph{near-linear work} and \emph{polylogarithmic depth}.

In 2018, Fineman \cite{fineman2018nearly} presented a breakthrough algorithm that, given a directed graph with $n$ vertices and $m$ edges, computes single-source reachability using $\tilde{O}(m)$ work and $\tilde{O}(n^{2/3})$ depth. This sparked a line of exciting work that  improved the depth to $n^{1/2+o(1)}$ \cite{LiuJS19}, extended the result to approximate shortest paths \cite{cao2020efficient}, and further extended it to exact shortest paths \cite{cao2023parallel,rozhovn2023parallel}.

The core graph-theoretic objects behind all these results are \emph{shortcuts} or \emph{hopsets}.
A set of directed edges $H\subseteq V\times V$ is an \emph{$h$-shortcut} for $G$ (of size $|H|$ and depth $h$) if $H$ is a subgraph of the transitive closure of $G$ and for every
$u,v\in V$ that are reachable in $G$, there is a path from $u$ to $v$ in the augmented graph $G\cup H$ that uses at most $h$ edges. 
A set of directed edges with weights $H\subseteq V\times V$ is an \emph{$(\alpha,h)$-hopset} for $G$ if for every $u,v\in V$, $\dist_G(u,v)\le \dist^{(h)}_{G\cup H}(u,v)\le \alpha\cdot \dist_G(u,v)$ where $\dist^{(h)}_{G\cup H}(u,v)$ is the length of the shortest path from $u$ to $v$ using at most $h$ edges.

Given these objects, parallel algorithms for single-source reachability and shortest paths are immediate. For example, given a shortcut $H$ for a graph $G$ of size $s$ and diameter $h$, we can compute single-source reachability by simply running breadth-first search on $G\cup H$ in $O(m+s)$ work and $O(h)$ depth. Thus, the parameters $s$ and $h$ determine the efficiency.
One can similarly obtain parallel algorithms for single-source shortest paths, given hopset using standard techniques too. 

Thus, efficient construction of shortcuts and hopsets in the parallel setting has been \emph{the} common goal of all prior results \cite{Fineman20,LiuJS19,cao2020efficient}. Towards this, they follow the same strategy. 
\begin{enumerate}
    \item \textbf{(Sequential)}: Show algorithms for building shortcuts or hopsets in near-linear work. 
    \item \textbf{(Parallel)}: Modify the sequential algorithms so that they also have small depth. 
\end{enumerate}

The key algorithmic ideas usually lie in the first step -- the straightforward construction all require super-linear work, and it is already very non-trivial to obtain a fast sequential algorithm. However, the second step is much more complicated and technical. For example, the result by Fineman \cite{Fineman20} spends 10 pages on the sequential construction (Section3), but then takes at least 20 pages for the parallel modification (Section 4-5). Similarly, \cite{LiuJS19} takes 8 pages for sequential construction (Section 4), but takes 15 more pages for parallelization (Section 5). 
In \cite{cao2020efficient}, even though the changes are less substantial, their hopset construction also requires some modifications to be implemented in the parallel setting.

Thus, it would be highly desirable if researchers could focus solely on the first step, which encapsulates the key ideas, and treat the second step as a black box, thereby avoiding all technicalities. 
Typically, the first step of these algorithms involves running BFS or shortest path algorithms on a subgraph of the input graph. 
We observe that this first step has low depth under the assumption that the input graph is \emph{shallow} (see \Cref{sec:simplification} for more details). 
This paper achieves precisely this goal. 
Formally, we establish the following theorem.

\begin{theorem}[Simplified and informal version of \Cref{col:shortcutreduction} (shortcut) and \Cref{lem:hopsetreduction} (hopset)]\label{thm:overview}
  Suppose there exists an algorithm $\cA_0$ that, given a directed graph of reachability diameter\footnote{The reachability diameter of a graph is defined as the maximum distance between any two reachable nodes.}  $\hO{h}$, computes an $h$-shortcut of $\hO{n}$ size.\footnote{Throughout the paper, we use $n,m$ to denote the number of vertices and edges of the input graph, and $\tO{\cdot}$ to omit $\polylog(n)$ factors and $\hO{m}$ to omit subpolynomial factors.} Then there exists an algorithm $\cA$ that makes $\tO{1}$ calls to $\cA_0$ with additional $\hO{m}$ work and $\hO{h}$ depth and computes an $h$-shortcut of $\hO{n}$ size for any directed graph. A similar result holds for hopsets.
\end{theorem}

This blackbox significantly simplifies all previous shortcut and hopset constructions, because once the input graph has a small reachability diameter, their sequential algorithms automatically have low depth. The second step in their presentation can be removed entirely. 
Indeed, our black-box will facilitate the design of parallel algorithms for reachability and shortest paths in the future, too.

\paragraph{Technical Overview.}

In this technical overview we focus on the \emph{shortcut} construction;  
the \emph{hopset} construction follows the same high‑level idea but involves additional details.

Our goal is to reduce the reachability diameter of a graph~$G$ with reachability diameter at least~$2h$ by a factor of~$2$.  
Once such a reduction is available, we simply invoke the algorithm for $\log n$ rounds, each time  
adding the newly computed shortcuts to~$G$ before the next round.  
After $\log n$ rounds the reachability diameter drops to at most~$h$.

The algorithm begins with a low‑diameter decomposition (LDD) of $G$ with the diameter parameter $\tT{h}$  
(here $\tT{h}$ abbreviates $h$ multiplied by a large enough polylogarithmic factor).  
The formal definition of an LDD is in \Cref{lem:LDD}.  
In short, the LDD is a randomized partition of $V(G)$ into strongly connected components
$V_1,\dots,V_z$ such that
\begin{itemize}
  \item each component $V_i$ has \emph{strong diameter}\footnote{The strong diameter of~$V_i$ is the diameter of the induced subgraph $G[V_i]$, which is the longest distance between \emph{every pair} of nodes, in contrast to reachability diameter.} $d=\tO{h}$, and
  \item for every edge $(u,v)$ the probability that it is \emph{reversed}—that is, $v\in V_i$ and $u\in V_j$ with $i<j$—is at most $O(\log^2 n/d)=o(1/h)$ (because $d=\tT{h}$).
\end{itemize}

Let $H$ denote the subgraph obtained from $G$ by deleting all reversed edges.

For every path $p$ connecting two vertices $s,t$ with $\ell$ edges we wish to add shortcut edges so that $s$ and $t$ are connected by a path of at most $\ell/2$ edges.  
According to the LDD guarantees, in expectation there are at most $o(\ell/h)$ \emph{reversed} edges on~$p$.  
Thus $p$ can be split into $k=o(\ell/h)$ segments $p_1,p_2,\dots,p_k$ such that each $p_i$ lies completely inside~$H$.  
If we can add shortcuts that reduce the reachability diameter of~$H$ to~$h$, each $p_i$ can be replaced by an path of length~$h$, so the entire path is replaced by one of length
\[
o(\ell/h)\cdot h \;=\; o(\ell) \;<\; \ell/2,
\]
achieving the desired factor‑$2$ reduction.

\medskip
\noindent
It therefore remains to add shortcuts that bring the reachability diameter of~$H$ down to~$h$.  
By construction, each $V_i$ produced by the LDD has strong diameter $\tO{h}$, exactly matching the precondition of the oracle~$\cA_0$.  
We call $\cA_0$ on every induced subgraph $G[V_i]$ to obtain a shortcut set $E_i$ so that $G[V_i]\cup E_i$ has reachability diameter at most~$h$.  
After inserting both $E_i$ and $E_{i+1}$, every two adjacent vertex sets $G[V_i\cup V_{i+1}]\cup E_i\cup E_{i+1}$ has reachability diameter~$2h$, again satisfying the conditions for~$\cA_0$.  
Hence we merge every two adjacent vertex sets and call $\cA_0$ on each merged graph.  
Repeating this procedure for $O(\log n)$ rounds ultimately ensures that $H$ attains reachability diameter at most~$h$.

\section{Preliminaries}

\paragraph{Asymptotic notations.} $O(\cdot)$ hides a global fixed large enough constant factor. $\tO{\cdot}$ hides a $\polylog(n)$ factor, where $n$ is always the input size instead of the value inside $\tO{\cdot}$. 

\paragraph{Graphs and edge lengths.} All graphs in this paper are directed graphs (digraphs) $G=(V,E)$, possibly with edge length $\ell_G:E\to\bbN^+$, or simply $\ell$ if $G$ is clear from the context. Throughout the paper, we use $N$ to denote the upper bound for the edge length. We assume $N$ is polynomial on $n$. A \emph{directed acyclic graph (DAG)} is a digraph without cycles. We require the edge lengths to be positive integers so the length of a path is an upper bound on the \emph{hoplength} $|p|$ - defined as the number of edges on the path. When we do not specify the length function of $G$, we assume $G$ has unit edge length $\ell(e)=1$ for every $e\in E$. 

\paragraph{Distances and diameters.} For $u,v\in V$, we say $u$ can \emph{reach} $v$, $v$ is \emph{reachable} from $u$, or $(u,v)$ is a \emph{reachable pair}, if there is a path from $u$ to $v$. The distance from $u$ to $v$ (denoted by $\dist_G(u,v)$) is defined as the minimum length of paths from $u$ to $v$, in case such a path does not exist, the distance is $\infty$. The \emph{\rd} of a directed graph is the maximum distance between every two reachable pairs. 

Given an integer $h$, the \emph{$h$-restricted distance} from $u$ to $v$, denoted by $\dist_G^{(h)}(u,v)$, is defined as the minimum length of paths from $u$ to $v$ with hoplength at most $h$. We say $G$ has \emph{$\alpha$-approximate shortest path hopbound $h$} if $\dist^{(h)}_G(u,v)\le \alpha\cdot \dist_G(u,v)$ for all $u,v$.

Given a vertex set $S\subseteq V$, the \emph{weak diameter} of $S$ on $G$ is the minimum distance on $G$ between every pair of nodes in $S$. The \emph{strong diameter} of $S$, or simply the diameter of $G[S]$, is the minimum distance on $G[S]$ between every pair of nodes in $S$.

\paragraph{Shortcut.} A shortcut for a graph $G=(V,E)$ is a set of edge $E'\subseteq V\times V$ such that every $(u,v)\in E'$ is a reachable pair in $G$. The \emph{hopbound} or \emph{depth} of $E'$ is defined as the \rd{} of $G\cup E':=(V,E\cup E')$. The size of $E'$ is defined as $|E'|$. 

\paragraph{Hopset.} A hopset for a graph $G=(V,E)$ with edge length is a set of edges $E'\subseteq V\times V$ with edge length such that for every $u,v\in V$, we have $\dist_G(u,v)\le \dist_{G\cup E'}(u,v)$. We say $E'$ has length slackness $\alpha$ and hopbound $\beta$ if $\dist^{(\beta)}_{G\cup E'}(u,v)\le \alpha\cdot \dist_G(u,v)$, in which case we call $E'$ a $(\alpha,\beta)$-hopset. The size of $E'$ is defined as $|E'|$.

\section{Reduction to Shallow Graphs}
\begin{theorem}\label{lem:hopsetreduction}
    Let $n,\lambda,h\in\bbN_{\ge 9}$ and $a,b,\eps\in\bbR_{\ge 0}$ be fixed values satisfying $\lambda>c_0\cdot \log^3n\cdot (1/\eps^2+1)$ where $c_0$ is a sufficiently large constant.
    
    Suppose there is an algorithm  $\cA_0$ that, given $\alpha_0$ and an arbitrary digraph $G_{0}$ with $n$ vertices and $\alpha_0$-approximate shortest path hopbound $\lambda h$, returns a $(\alpha_0\cdot(1+\eps),h)$-hopset of size $am_0+b$ where $m_0$ is the number of edges in $G_0$.\footnote{The only variables here are $\alpha_0$, number of edges in $G_0$ and the graph structure of $G_0$; the other parameters like $n,\lambda,h,a,b,\eps$ are fixed values. Such a setting makes the theorem easy to state (for example, in reality, $a,b$ could be depending on $n,\eps$, but since they are fixed values here, we do not need to write them as functions)}

    Then there is a randomized algorithm $\cA$ that, given a digraph $G$ with $n$ vertices, returns a $((1+\eps)^{O(\log^2_{\lambda}n)},h)$-hopset of size $S(m)\le n^{o(1)}\cdot (1+a)^{O\!\bigl(\log_{\lambda}^{2}n\bigr)}\,(m + b)$ with high probability. Moreover, if $a<1/(c_0\log^2_\lambda n)$, then  $S(m)\le \tO{am+b}$.
    
    The work and depth of $\cA$ is bounded by $\tO{1}$ (sequential) calls to $\cA_0$ on graphs with at most $S(m)+m$ edges, and an additional $\tO{m}$ work and $\tO{\lambda h}$ depth.
\end{theorem}

If we only care about reachability with unlimited length slackness, we can set $\eps$ to be a sufficiently large value and immediately get the following result.

\begin{corollary}\label{col:shortcutreduction}
    Let $n,\lambda,h\in\bbN_{\ge 9}$ and $a,b\in\bbR_{\ge 0}$ be fixed values satisfying $\lambda>c_0\log^3 n$ where $c_0$ be a sufficiently large constant.
    
    Suppose there is an algorithm  $\cA_0$ that, given $\alpha_0$ and an arbitrary digraph $G_{0}$ with $n$ vertices and \rd{} $\lambda h$, returns a shortcut of hopbound $h$ and size $am_0+b$ where $m_0$ is the number of edges in $G_0$.

    Then there is a randomized algorithm $\cA$ that, given a digraph $G$ with $n$ vertices, returns a shortcut of hopbound $h$ and size $S(m)\le n^{o(1)}\cdot (1+a)^{O\!\bigl(\log_{\lambda}^{2}n\bigr)}\,(m + b)$ with high probability. Moreover, if $a<1/(c_0\log^2_\lambda n)$, then  $S(m)\le \tO{am+b}$.
    
    The work and depth of $\cA$ is bounded by $\tO{1}$ (sequential) calls to $\cA_0$ on graphs with at most $S(m)+m$ edges, and an additional $\tO{m}$ work and $\tO{\lambda h}$ depth.

\end{corollary}
    
\subsection{Reducing Clustered DAGs to Shallow Graphs}

\begin{definition}
    A directed graph $G$ with edge lengths is a $d$-clustered DAG if every strongly connected component of $G$ has strong diameter $d$. 
\end{definition}

\begin{lemma}\label{lem:hopsetreductionDAG}
    Let $n,\eps,\lambda,h,a,b,\cA_0$ be defined in \Cref{lem:hopsetreduction}. There is an algorithm $\cAD$ that, given a $\lambda h$-clustered digraph $G$ with positive integer edge length and $n$ vertices, and the topological order of every strongly connected component of $G$ as inputs, returns a $((1+\eps)^{2\log_\lambda n},h)$-hopset of size $S(m)\le O((1+a)^{O(\log_\lambda n)}\cdot (m+b\log_{\lambda}n))$. Moreover, if $a=o(1/\log_{\lambda}n)$, we further have $S(m)\le O(\log_\lambda n\cdot(am+b))$.
    
    The work and depth of $\cAD$ is bounded by $O(\log_{\lambda}n)$ sequential calls to $\cA_0$ on graphs with at most $O(S(m)+m)$ edges.
\end{lemma}

\paragraph{Algorithms.} Let the strongly connected components of $G$ following the topological order be $(C_1,...,C_z)$. The algorithm will maintain a graph $G^{(i)}$ and a sequence of vertex sets throughout the algorithm $(V^{(i)}_1,...,V^{(i)}_{z^{(i)}})$. Initially, we set $G^{(0)}=G$ and $z^{(0)}=z$ and $V^{(0)}_j=C_j$ for every $j\in[z]$. 

The algorithm repeats for $O(k)$ iterations as follows. In the $i$-th iteration, the algorithm apply $\cA_0$ with inputs $\alpha_0=\alpha^{(i)}_0:=(1+\eps)^{i-1}$ and the graph 

\[\tG^{(i-1)}:=G^{(i-1)}[V^{(i-1)}_1]\cup G^{(i-1)}[V^{(i-1)}_2]\cup ...\cup G^{(i-1)}[V^{(i-1)}_{z^{(i-1)}}],\]
i.e., the graph after deleting edge crossing different $V^{(i-1)}_j$ in $G^{(i-1)}$, to get a hopset $E^{(i)}$ (we will argue the validity of this oracle call in the analysis). 
Next, define
\[G^{(i)}=G^{(i-1)}\cup E^{(i)}.\]
\begin{itemize}
    \item If $z^{(i-1)}=1$, then the algorithm stops and returns $\cup_{j\le i}E^{(j)}$.
    \item Otherwise, the algorithm constructs 
\[V^{(i)}_j=\cup_{\lambda(j-1) < j'\le \lambda j}V^{(i-1)}_{j'}\]
and continue the iterations. Note that $z^{(i)} = \lceil z^{(i-1)}/\lambda \rceil$.
\end{itemize}

\paragraph{Correctness (validity of oracle calls).} We first show that the oracle calls of the algorithm satisfy the condition for $\cA_0$ described in \Cref{lem:hopsetreduction}. We will prove it by induction on $i$. In the first iteration, $G^{(i-1)}=G^{(0)}=G$, $V^{0}_j=C_j$, since $G$ is a $\lambda h$-clustered digraph, the strong diameter of $C_j$ is at most $\lambda h$ for all $j\in[z]$. Thus, $G[C_1]\cup G[C_2]\cup...\cup G[C_z]$ has $1$-approximate shortest path hopbound at most $\lambda h$, satisfying the condition for the oracle call to $\cA_0$. 

At the end of the $i$-th iteration, according to the output condition of $A_0$, for all $j$, $G^{(i)}[V^{(i-1)}_j]$ has $\alpha^{(i)}_0\cdot (1+\eps)$-approximate shortest path hopbound at most $h$ because $G^{(i)}=G^{(i-1)}\cup E^{(i)}$ where $E^{(i)}$ is a $(\alpha^{(i)}_0\cdot (1+\eps),h)$-hopset. Remember the definition $V^{(i)}_j=\cup_{\lambda(j-1)<j'\le \lambda j}V^{(i-1)}_{j'}$, so every path in $G^{(i)}[V^{(i)}_j]$ can cross at most $\lambda$ different $V^{(i-1)}_j$ following the topological order. Thus, $G^{(i)}[V^{(i)}_j]$ has $\alpha^{(i)}_0\cdot (1+\eps)$-approximate shortest path hopbound at most $\lambda h$. Remember that $\alpha^{(i+1)}_0=\alpha^{(i)}_0\cdot (1+\eps)$, so the oracle call to $\cA_0$ in the $i+1$-the iteration is valid. The induction holds.

\paragraph{Correctness (hopset quality).} Now we argue that the algorithm output a valid $((1+\eps)^{2\log_\lambda n},h)$-hopset. In the last iteration, $\tG^{(i-1)}=G^{(i-1)}[V^{(i-1)}_1]$, which means $V^{(i-1)}_1=V$. According to the validity of the oracle calls we just proved, we have that $G^{(i)}$ has $\alpha^{(i)}_0\cdot (1+\eps)$-approximate shortest path hopbound at most $h$. Remember that $G^{(i)}=G\cup(\cup_{j\le i}E_i)$, so the returned set $\cup_{j\le i}E_i$ is a $((1+\eps)^i,h)$-hopset. Notice that $i$ can be at most $2\log_\lambda n$ since each iteration shrinks $z^{(i)}$ by a factor of at least $\lambda-1$ (also recall that $\lambda\ge 9$). 

\paragraph{Correctness (hopset size).} Now we calculate the size of the hopset. In the $i$-th iteration, we apply $\cA_0$ on a subgraph of $G^{(i-1)}$, which has number of edges at most $m+\sum_{j< i}|E_j|$. Thus, we get the recursive inequality $|E_i|\le a\cdot (m+\sum_{j< i}|E_j|)+b$. Solving this inequality gives $\sum_{j}|E_j|\le O((1+a)^{O(\log_\lambda n)}\cdot (m+b\log_{\lambda}n))$; Moreover, if $a=o(1/\log_\lambda n)$, we get $\sum_{j}|E_j|\le O(\log_\lambda n\cdot(am+b))$.

\paragraph{Complexity.} The algorithm consists of $O(\log_\lambda n)$ calls to $\cA_0$, each call is on a graph obtained as a subgraph of $G$ applied by the final returned hopset, which has a number of edges at most $O(S(m)+m)$. This cost dominates other costs of the algorithm.

\subsection{Reducing General Graphs to Clustered DAGs}

\paragraph{Algorithm.} Let $\lambda'=\frac{\eps \cdot\lambda}{\sqrt{c_0}\cdot \log^2 n}$ (where $c_0$ is a sufficiently large constant, when $\eps>1$, we let $\lambda'=\frac{\lambda}{\sqrt{c_0}\log^2n}$). The algorithm runs for $O(\log_{\lambda'}n)$ epochs. Initially we set $G^{(0)}=G$. The $i$-th epoch will construct a hopset $E^{(i)}$ and define $G^{(i)}=G^{(i-1)}\cup E^{(i)}$. Roughly speaking, $E^{(i)}$ is intended to bring down the approximate shortest path hopbound of $G^{(i-1)}$ by a factor of $\lambda'$.

The $i$-th epoch consists of $O(\log N)$ phases, each phase will construct a hopset $E^{(i)}_j$, and the final hopset for the $i$-th epoch is defined as $E^{(i)}:=\cup_{j}E^{(i)}_j$. Roughly speaking, $E^{(i)}_j$ will only shortcut the paths whose length-to-hopbound ratio are between $[2^j,2^{j+1}]$. 

\paragraph{The algorithm for one phase.} Define the scaled down length function 
\[\ell_j(e):=\left\lceil\ell(e)\cdot \min\left(1,\frac{1}{\eps\cdot 2^j}\right)\right\rceil\]

We will use the following low-diameter decomposition algorithm.

\newcommand{\rem}{\mathrm{rem}}
    \begin{lemma}[Proved in \Cref{sec:ldd}]\label{lem:LDD}
        Given a directed graph $G=(V,E)$ with positive integral edge length and a the diameter parameter $d$, we can compute a set of edges $E^{\rem}$ and a list of vertex sets $V_1,...,V_{\ell}$ such that
        \begin{enumerate}
            \item $V_1,...,V_{\ell}$ are the strongly connected components of $G-E^{\rem}$ following the topological ordering, i.e., every edge in $G-E^{\rem}$ can only point from $V_i$ to $V_j$ for $j\ge i$,
            \item each $V_i$ has weak diameter at most $d$,
            \item for every $e\in E$, we have $\Pr[e\in E^{\rem}]=O(\log^2n/d)$.
        \end{enumerate}

        The algorithm has work $\tO{m}$ and depth $\tO{d}$.
    \end{lemma}

Notice that \Cref{lem:LDD} has random output. We will repeat $O(\log n)$ times the following procedure, each repetition will generate a random $E^{(i)}_j$, and the final $E^{(i)}_j$ will be the union of them. To avoid ambiguity, we use $\bar{E}^{(i)}_j$ to denote one sample, and $E^{(i)}_j$ is just the union of all $\bar{E}^{(i)}_j$.

We apply \Cref{lem:LDD} on the graph $G$ with edge length $\ell_j$ and diameter parameter $\lambda h/2$ to get $E^{\rem}_j$ along with the topologically ordered strongly connected components of $G-E^{\rem}_j$. Then, for each strongly connected component $C$ of $G-E^{\rem}_j$, we pick an arbitrary node $u\in C$ and connect $u$ and every other vertex in $C$ by a bidirected edge pair with length $\lambda h/2$. We call these edges as \emph{stars}. The graph $G-E^{\rem}_j$ union the stars is passed to the algorithm in \Cref{lem:hopsetreductionDAG} (the validity of this call is easy to see: the stars force each strongly connected component of $G-E^{\rem}_j$ to have strong diameter at most $\lambda h$). \Cref{lem:hopsetreductionDAG} outputs a hopset. Denote the union of the output hopset and the stars by $\tE^{(i)}_j$. To get $\bar{E}^{(i)}_j$, we scale up the edge length of $\tE^{(i)}_j$ by a factor of $\lceil \eps\cdot 2^j\rceil$. 

\paragraph{Correctness (hopset quality).} We will prove the following lemma, which shows the $j$-th phase indeed shortcut path with length-to-hopbound ratio between $[2^j,2^{j+1}]$.

\begin{lemma}\label{lem:onephase}
    For every $i,j\in \bbN^+$, fix two vertices $u,v\in V$ and a path $p$ in $G^{(i-1)}$ from $u$ to $v$, if the length-to-hopbound ratio of $p$ is between $[2^j,2^{j+1}]$, then $\dist^{|p|/\lambda'+h}_{G^{(i-1)}\cup E^{(i)}_j}(u,v)\le (1+\eps)^{3\log_\lambda n}\cdot \ell(p)$.
\end{lemma}

\begin{proof}
    Let $\sigma=\lceil\eps\cdot 2^j\rceil$. We first show that the scaled down length $\ell_j(p)$ is close to $\ell(p)/\sigma$. 
    \begin{lemma}\label{lem:apxellj}
        We have $\ell(p)/\sigma\le \ell_j(p)\le (1+2\eps)\cdot(\ell(p)/\sigma)$.
    \end{lemma}
    \begin{proof}
        Firstly, if $\eps\cdot 2^j\le 1$, then $\ell_j(e)=\ell(e)=\ell(e)/\sigma_j$. Otherwise, we have $2^j>1/\eps$. Define $\tilde{\ell}_j:=\ell/(\eps\cdot 2^j)$, then $\ell_j(e)$ is rounding up the edge length of $\tilde{\ell}$ to the closest integer. Notice that $|p|\le \ell(p)/2^j$, so we get
    \[\ell_j(p)\le \tilde{\ell}_j(p)+|p|\le \ell(p)/\sigma+\ell(p)/2^j\le \ell(p)/\sigma+2\eps\cdot \ell(p)/\sigma=(1+2\eps)\cdot(\ell(p)/\sigma)\]
        Moreover, we have
    \[\ell_j(p)\ge \tilde{\ell}_j(p)\ge \ell(p)/\sigma.\qedhere\]
    \end{proof}

    According to the guarantee of low-diameter decomposition \Cref{lem:LDD}, we have
    \[\bbE\left[|E(p)\cap E^{\rem}_j|\right]=O\left(\frac{\log^2 n}{\lambda h}\right)\cdot \ell_j(p).\]

    By markov inequality and the fact that random bits for all the $O(\log n)$ generated low-diameter decomposition are independent, the following event should happen with high probability for one of the samples:
    \[|E(p)\cap E^{\rem}_j|\le O\left(\frac{\log^2 n}{\lambda h}\right)\cdot \ell_j(p).\]

    Consider the $|E(p)\cap E^{\rem}_j|+1$ segments of $p$ obtained by cutting all the edges in $E(p)\cap E^{\rem}_j$, each of these segments is completely contained in $G-E^{\rem}_j$, and thus can be replaced by a path in $(G-E^{\rem}_j)\cup E^{(i)}_j$ with hoplength $h$ and approximation ratio $(1+\eps)^{2\log_\lambda n}$ according to \Cref{lem:hopsetreductionDAG}. Denote the path after substituting all the segments by the hoplength $h$ path by $p'$, the length increase by at most a factor of $(1+\eps)^{2\log_\lambda n}$, and the hoplength satisfies
    \begin{align*}
        |p'|&\le h\cdot O\left(\frac{\log^2 n}{\lambda h}\right)\cdot \ell_j(p)+h\\
        &\le O\left(\frac{\log^2 n}{\lambda}\right)\cdot (1+2\eps)\cdot(\ell(p)/\sigma)+h\\
        &\le O\left(\frac{\log^2 n}{\eps \cdot\lambda}\right)\cdot \ell(p)/2^j+h\\
        &\le O\left(\frac{\log^2 n}{\eps \cdot\lambda}\right)\cdot|p|+h.
    \end{align*}
    The second inequality is because of \Cref{lem:apxellj}, the last inequality is because the length-to-hopbound ratio of $p$ is between $[2^j,2^{j+1}]$. Moreover, after scaling up the length from $\ell_j$ to $\ell$, the length of $p'$ satisfies
    \[\ell(p')\le \sigma\cdot \ell_j(p')\le \sigma\cdot(1+\eps)^{2\log_\lambda n}\cdot \ell_j(p)\le (1+\eps)^{2\log_\lambda n}\cdot(1+2\eps)\cdot\ell(p)\]
    This finishes the proof.
\end{proof}

Given \Cref{lem:onephase}, we can check the hopset quality. Fix one pair $u,v$ such that the shortest path from $u$ to $v$ in $G$ is $p$. We define $p^{(0)}=p$. In the $i$-th epoch, according to \Cref{lem:onephase}, since there must be one phase that the length-to-hopbound ratio of $p$ is between $[2^j,2^{j+1}]$, there must be a path $p^{(i)}$ with hoplength $|p^{(i-1)}|/\lambda'+h$ and length $(1+\eps)^{3\log_\lambda n}\cdot\ell(p^{(i-1)})$ in $G^{(i-1)}\cup E^{(i)}_j$. Since there are $O(\log_{\lambda'} n)$ epochs, in the end, there is a path from $u$ to $v$ in the final graph with all the hopset edges with hoplength $h$ and length $(1+\eps)^{O(\log_\lambda n\cdot \log_{\lambda'}n)}\cdot\dist_G(u,v)$. This proves one directed of the hopset quality.

We also need to make sure the hopset do not decrease the distance. They are relatively easy to check. The stars does not increase the distance because it is added to strongly connected components with bounded weak diameter. Other hopset edges are added by applying \Cref{lem:hopsetreductionDAG} to the scaled-down graph, and scaling up by at least same factor. Thus, they also do not decrease the distance.

\paragraph{Correctness (hopset size).} Now we calculate the size of the hopset. The $i$-th epoch contains $O(\log N)$ phases, where each phase contains $O(\log n)$ random low-diameter decomposition with oracle calls to \Cref{lem:hopsetreductionDAG}, so in each epoch the algorithm applies \Cref{lem:hopsetreductionDAG} for $O(\log N\cdot \log n)$ many times in total on $G^{(i-1)}$, which contains at most $m+\sum_{j<i}E^{(j)}$ edges. Thus, we get the recursive inequality 
\[|E^{(i)}|\le \log N\cdot \log n\cdot (1+a)^{O(\log_\lambda n)}\cdot (m+\sum_{j<i}E^{(j)}+b\log_{\lambda}n).\]
Solving this inequality gives $\sum_{i\le O(\log_{\lambda'}n)}E^{(i)}\le O\left((\log N\log n)^{O(\log_{\lambda'}n)}\cdot (1+a)^{O(\log_\lambda n\log_{\lambda'}n)}\cdot (m+b)\right)$.
Moreover, for small $a=o(1/\log_\lambda n\log_{\lambda'}n)$, we have the recursive inequality $|E^{(i)}|\le O(\log n\log N\log_\lambda n\cdot (a(m+\sum_{j<i}E^{(j)})+b)$, solving this inequality gives $\sum_{i\le O(\log_{\lambda'}n)}E^{(i)}\le O(\log n\log N\log_{\lambda}n\log_{\lambda'}n)(am+b)$.

\paragraph{Complexity.} As explained in the last paragraph, each epoch of the algorithm contains $O(\log N\log n)$ many calls to \Cref{lem:hopsetreductionDAG} (where there are $\log N$ phases, each phase contains $O(\log n)$ repeated calls to sample low diameter decomposition). There are $O(\log_{\lambda'} n)$ epochs, and each oracle call to \Cref{lem:hopsetreductionDAG} makes $O(\log_\lambda n)$ calls to $\cA_0$, so in total,
the algorithm consists of $O(\log n\log N\log_{\lambda'}n\log_{\lambda}n)$ calls to $\cA_0$, each on a graph with number of edges bounded by $m$ plus the total hopset size. Besides the calls to $\cA_0$, the algorithm additionally uses $O(\log n\log N\log_{\lambda'}n)$ calls to low-diameter decomposition, which has complexity specified in \Cref{lem:LDD}.

\paragraph{Simplifications.} Since the hopset size bound and complexity looks scary, we simplify it by assuming $N$ is polynomial on $n$ and use $\tO{\cdot}$ to hide polylogarithmic factors. We further assume a larger $\lambda >c_0\cdot \log^3n(1/\eps^2+1)$ to make sure $\log_{\lambda }n$ and $\log_{\lambda'}n$ are within the same order.
After these simplifications, the hopset size becomes 
\[
\mathrm{HopSetSize} =
\begin{cases}
n^{o(1)}\,(1+a)^{O\!\bigl(\log_{\lambda}^{2}n\bigr)}\,(m + b),
& \text{for general }a, \\[1ex]
\tO{a\,m + b}
& \text{if }a = \dfrac{1}{c_{0}\,\log_{\lambda}^{2}n}\text{ for sufficiently small }c_{0}.
\end{cases}
\]

The algorithm makes $\tO{1}$ calls to $\cA_0$ with additional $\tO{m}$ work and $\tO{\lambda h}$ depth.

\bibliographystyle{alpha}
\bibliography{refs}

\newcommand{\etalchar}[1]{$^{#1}$}
\begin{thebibliography}{RHM{\etalchar{+}}23}

\bibitem[ABC{\etalchar{+}}24]{AshvinkumarBCGH24}
Vikrant Ashvinkumar, Aaron Bernstein, Nairen Cao, Christoph Grunau, Bernhard Haeupler, Yonggang Jiang, Danupon Nanongkai, and Hsin{-}Hao Su.
\newblock Parallel, distributed, and quantum exact single-source shortest paths with negative edge weights.
\newblock In {\em {ESA}}, volume 308 of {\em LIPIcs}, pages 13:1--13:15. Schloss Dagstuhl - Leibniz-Zentrum f{\"{u}}r Informatik, 2024.

\bibitem[CF23]{cao2023parallel}
Nairen Cao and Jeremy~T Fineman.
\newblock Parallel exact shortest paths in almost linear work and square root depth.
\newblock In {\em Proceedings of the 2023 Annual ACM-SIAM Symposium on Discrete Algorithms (SODA)}, pages 4354--4372. SIAM, 2023.

\bibitem[CFR20]{cao2020efficient}
Nairen Cao, Jeremy~T Fineman, and Katina Russell.
\newblock Efficient construction of directed hopsets and parallel approximate shortest paths.
\newblock In {\em Proceedings of the 52nd Annual ACM SIGACT Symposium on Theory of Computing (STOC)}, pages 336--349, 2020.

\bibitem[Fin18]{fineman2018nearly}
Jeremy~T Fineman.
\newblock Nearly work-efficient parallel algorithm for digraph reachability.
\newblock In {\em Proceedings of the 50th Annual ACM Symposium on Theory of Computing (STOC)}, pages 457--470, 2018.

\bibitem[Fin20]{Fineman20}
Jeremy~T. Fineman.
\newblock Nearly work-efficient parallel algorithm for digraph reachability.
\newblock {\em {SIAM} J. Comput.}, 49(5), 2020.

\bibitem[JLS19]{LiuJS19}
Arun Jambulapati, Yang~P. Liu, and Aaron Sidford.
\newblock Parallel reachability in almost linear work and square root depth.
\newblock In {\em Proceedings of the 60th Annual {IEEE} Symposium on Foundations of Computer Science ({FOCS})}, pages 1664--1686, 2019.

\bibitem[Kle93]{Klein93}
Philip~N. Klein.
\newblock Parallelism, preprocessing, and reachability: {A} hybrid algorithm for directed graphs.
\newblock {\em J. Algorithms}, 14(3):331--343, 1993.

\bibitem[KS97]{KleinS97}
Philip~N. Klein and Sairam Subramanian.
\newblock A randomized parallel algorithm for single-source shortest paths.
\newblock {\em J. Algorithms}, 25(2):205--220, 1997.

\bibitem[RHM{\etalchar{+}}23]{rozhovn2023parallel}
V{\'a}clav Rozho{\v{n}}, Bernhard Haeupler, Anders Martinsson, Christoph Grunau, and Goran Zuzic.
\newblock Parallel breadth-first search and exact shortest paths and stronger notions for approximate distances.
\newblock In {\em Proceedings of the 55th Annual ACM Symposium on Theory of Computing (STOC)}, pages 321--334, 2023.

\bibitem[Spe97]{Spencer97}
Thomas~H. Spencer.
\newblock Time-work tradeoffs for parallel algorithms.
\newblock {\em J. {ACM}}, 44(5):742--778, 1997.

\bibitem[UY91]{UllmanY91}
Jeffrey~D. Ullman and Mihalis Yannakakis.
\newblock High-probability parallel transitive-closure algorithms.
\newblock {\em {SIAM} J. Comput.}, 20(1):100--125, 1991.

\end{thebibliography}
\appendix    
\section{Parallel Directed Low-Diameter Decomposition}\label{sec:ldd}
We will use the LDD algorithm in~\cite{AshvinkumarBCGH24}. However, we need to adjust their algorithm to fit our case for two reasons (i) their algorithm is stated as a reduction to SSSP, which is far from our goal of $\tO{m}$ work and $\tO{d}$ depth as the current best almost linear work algorithm for SSSP has $\sqrt{n}$ depth, (ii) compared to their output, we additionally require the output $V_1,...,V_{\ell}$ following the topological order, this is not described in their algorithm. We will show that both points are implicitly addressed in their algorithm.

    For the first point, we can examine all the places where the LDD algorithm described in \cite{AshvinkumarBCGH24} calls the SSSP subroutine (see \Cref{alg:lowdiamterdecomposition}, mainly the places of computing $\text{Ball}(v,d)$ in their paper, which means all the nodes with distance $d$ to $v$). All of them only call SSSP to find vertices of distance at most $\tO{d}$ to a specific vertex. This can be replaced by a simple BFS with depth $\tO{d}$ (even if the graph is weighted, we can view each integer weighted edge as a concatenation of unweighted edges and do BFS as usual), so the first point is solved.

    For the second point, we need to examine how $E^{rem}$ is constructed. To summarize their algorithm, $E^{rem}$ is constructed by a constant number of outgoing edges of vertex sets (which means edges $(u,v)$ where $u\in V',v\not\in V'$ for a specific vertex set $V'$), and their algorithm then recurse on each vertex sets. The topological order of these vertex sets in $G-E^{rem}$ can be computed in low depth because there are only a constant number of them. After ordering the constant number of vertex sets, we can further decompose these vertex sets in the recursive calls and follow the topological order, so the second point is solved. 
    
    For completeness, we write the whole algorithm which adds out changes as below (\Cref{alg:lowdiamterdecomposition}). The algorithm in addition outputs a list of vertex sets denoted as $\cV$. On the first return (line 16), there are no edges from $\cV_2$ to $\cV_1$ when $*=in$ because edges going into $A_{in}$ are included in $E^{rem}$ (same when $*=out$). On the second return (line 22), there are not edges to $\cV_1$ because edges going in to $A_{in}$ are included in $E^{rem}$; there are no edges going out from $\cV_2$ because they are included in $E^{rem}$ as well. This completes the correctness proof of the algorithm returns the correct topological order of $G-E^{rem}$. Other correctness and running time analysis follows directly from \cite{AshvinkumarBCGH24} since we do not change the base algorithm.

\newcommand{\Vout}{\ensuremath{V_{out}}\xspace}
\newcommand{\Vin}{\ensuremath{V_{in}}\xspace}
\newcommand{\Vheavy}{\ensuremath{V_{heavy}}\xspace}
\newcommand{\Ain}{\ensuremath{A_{in}}\xspace}
\newcommand{\Aout}{\ensuremath{A_{out}}\xspace}
\newcommand{\outcut}{\ensuremath{\delta^{+}}\xspace}
\newcommand{\incut}{\ensuremath{\delta^{-}}\xspace}
\newcommand{\Eref}{\ensuremath{E^{rem}}\xspace}
\newcommand{\Erem}{\ensuremath{E^{rem}}\xspace}
\newcommand{\Eneg}{\ensuremath{E^{neg}}\xspace}
\newcommand{\EremI}{\ensuremath{E^{rem}_1}\xspace}
\newcommand{\EremII}{\ensuremath{E^{rem}_2}\xspace}
\newcommand{\Eremin}{\ensuremath{E^{rem}_{in}}\xspace}
\newcommand{\Eremout}{\ensuremath{E^{rem}_{out}}\xspace}
\newcommand{\Erems}{\ensuremath{E^{rem}_{*}}\xspace}

\newcommand{\ingball}{\operatorname{Ball_G^{in}}}
\newcommand{\outgball}{\operatorname{Ball_G^{out}}}
\newcommand{\inball}{\operatorname{Ball^{in}}}
\newcommand{\outball}{\operatorname{Ball^{out}}}
\newcommand{\sgball}{\operatorname{Ball_G^{*}}}
\newcommand{\ball}{\operatorname{Ball}}

In this algorithm, $\mathcal{O}^{NN-SSSP}(G,v)$ refers to the oracle call to the negative weight shortest path algorithm, for which we use the simple BFS procedure. For more details of this algorithm, we refer insterested readers to \cite{AshvinkumarBCGH24}.
\begin{algorithm}[H]
	\caption{$E^{rem},{\color{red}\cV}\leftarrow LowDiameterDecomposition(G,d)$}\label{alg:lowdiamterdecomposition}
 
        \BlankLine
        If $G$ is an empty graph, return $\emptyset$\;
        Let $n$ and $c$ be defined as in full paper;
        \tcp*[h]{\textcolor{blue}{Phase 1: mark vertices as light or heavy}}
        
	Sample $\lceil c\log n\rceil$ vertices in $V$ uniformly at random, denoted as $S$\;\label{LDDline2}
	For each $v\in S$, use $\mathcal{O}^{NN-SSSP}(G,v)$ to find $\ingball(v,d/4)$ and $\outgball(v,d/4)$\;\label{line:ldd_s-balls}
 
        For each $v \in V$, compute $\ingball(v,d/4) \bigcap S$ and $\outgball(v,d/4) \bigcap S$ using Line~\ref{line:ldd_s-balls}\;

            \ForEach(\label{line:ldd-marking-loop}){$v \in V$}{
            
            If $|\ingball(v,d/4) \bigcap S| \leq .6 |S|$, mark $v$ \emph{in-light} \tcp*[f]{\textcolor{gray}{whp $|\ingball(v,d/4)| \leq .7|V(G)|$}}
                
            Else if $|\outgball(v,d/4) \bigcap S| \leq .6|S|$, mark $v$ \emph{out-light}\tcp*[f]{\textcolor{gray}{whp $|\outgball(v,d/4)| \leq .7|V(G)|$}}
                
            Else mark $v$  \emph{heavy} \tcp*[f]{\textcolor{gray}{whp $|\ingball(v,d/4)| > .5|V(G)|$ and $|\outgball(v,d/4)| > .5|V(G)|$}}
                
            }
	
        \BlankLine
        \tcp*[h]{\textcolor{blue}{Phase 2: creates sub-problems with small sizes}}
        
        Denote the set of \emph{in-light} vertices by $\Vin$, the set of \emph{out-light} vertices by $\Vout$\;
	$\Ain\leftarrow FindBalancedSet(G,\Vin,d,in)$, $\Eremin \leftarrow \incut(\Ain)$\;\label{LDDline4}
	$\Aout\leftarrow FindBalancedSet(G,\Vout,d,out)$, $\Eremout \leftarrow\outcut(\Aout)$\;\label{LDDline5}

        \tcp*[h]{\textcolor{blue}{Case 1: One of $\Ain,\Aout$ is balanced.}}
        
	\If{$A_*$ ($*$ can be $in$ or $out$) has size between $.1|V|$ and $.9|V|$}
        {
            $\EremI,{\color{red}\cV_1}\leftarrow LowDiameterDecomposition(G[A_*],d)$\; \label{recurse1}
            $\EremII,{\color{red}\cV_2} \leftarrow LowDiameterDecomposition(G[V\backslash A_*],d)$\;\label{recurse2}
            \Return{$\Erems\bigcup \EremI \bigcup \EremII$,{\color{red}$(\cV_1,\cV_2)$ if $*=in$, or $(\cV_2,\cV_1)$ if $*=out$}}\;\label{return1}
        }

        \tcp*[h]{\textcolor{blue}{Clean up: Check that $V\backslash(\Ain\bigcup \Aout)$ have small weak diameter.}}
        
	Pick an arbitrary vertex $u\in V\backslash(\Ain\cup \Aout)$. Use $\mathcal{O}^{NN-SSSP}(G,u)$ to find $\ingball(u,d/2),\outgball(u,d/2)$\;
        
        \If{$V\backslash(\Ain\cup \Aout)\not\subseteq \ingball(u,d/2)\bigcap\outgball(u,d/2)$ or $|\Ain\cup \Aout|\ge .5|V|$}
        { \Return $E$}

        \tcp*[h]{\color{blue}Case 2: both $\Ain,\Aout$ are small.}
        
        $\EremI,{\color{red}\cV_1}\leftarrow LowDiameterDecomposition(G[\Ain],d)$\;\label{recurse3}
        $\EremII,{\color{red}\cV_2}\leftarrow LowDiameterDecomposition(G[\Aout\backslash \Ain],d)$\;\label{recurse4}
            
	\Return{$\Eremin \bigcup \Eremout \bigcup \EremI \bigcup \EremII $,{\color{red}$(\cV_1,V-A_{in}-A_{out},\cV_2)$ }}\;\label{return3}
\end{algorithm}
\begin{algorithm}[H]
\caption{$A\leftarrow FindBalancedSet(G,V',d,*)$}\label{alg:smallsizedecomposition}
\KwData{Non-negative weighted directed graph $G=(V,E,w)$, a vertex set $V'\subseteq V$ and an integer $d$ satisfying $|\sgball(v,d/4)|\le .7|V|$ for any $v\in V'$.}
\KwResult{A set of vertices $A\subseteq V$ which have a Light Boundary and is either Balanced or contains $V'$.}
Suppose $V'=\{v_{1},v_{2},...,v_{\ell}\}$. Each vertex $v_{i}$ samples $d_i\sim GE[\min\left((c\log n)/d,1\right)]_{\le \lfloor d/4\rfloor}$ 
(A geometric random variable truncated to $\lfloor d/4\rfloor$)
\;\label{SSline1}
Find the smallest $i\in[\ell]$ such that $\left|\cup_{j\le i}\sgball(v_j,d_j)\right|>0.1|V|$, denoted as $k$. If $k$ does not exist, i.e. $\left|\cup_{j\in[\ell]}\sgball(v_j,d_j)\right|\le 0.1|V|$, set $k=\ell$. (Implementation in full paper)\;
 \label{SSline2}

\Return $\cup_{j\le k}\sgball(v_{j},d_j)$, \;\label{SSline3}
\end{algorithm}

\section{Simplification of Parallel Reachability and Shortest Path}\label{sec:simplification}

Assume the input graph $G$ has reachability diameter $\hO{h}$. We show that the sequential algorithm in \cite{LiuJS19} has depth $\hO{h}$ and produces a shortcut with $\hO{n} = \tO{0 \cdot m + n}$ edges and $\hO{h}$ depth. According to \Cref{col:shortcutreduction}, this immediately yields a reachability algorithm with $\hO{m}$ work and $\hO{\sqrt{n}}$ depth on general graphs by plugging in $\lambda = \tO{1}, a = 0, b = n$.

The only high-depth component of \cite{LiuJS19} is the BFS procedure on subgraphs of $G$. Each subgraph is defined by a set of pivot vertices $S$, and consists of all vertices that share the same reachability condition (reachable to, reachable from, or not reachable at all) with respect to all vertices in $S$. Observe that for any two vertices $u, v$ in such a subgraph, there exists a path $P$ in $G$ of length $\hO{h}$ connecting them. We argue that $P$ is entirely contained within the subgraph, implying that the subgraph has reachability diameter at most $\hO{h}$, and thus BFS on it has depth $\hO{h}$.

To see this, note that every vertex on $P$ has the same reachability condition to $S$ as $u$ and $v$. Specifically, for any $w \in S$:
\begin{itemize}
    \item If $w$ can reach both $u$ and $v$, then it can reach every vertex on $P$.
    \item If $w$ can be reached from both $u$ and $v$, then every vertex on $P$ can reach $w$.
    \item If $w$ can neither reach nor be reached from $u$ and $v$, then the same holds for all vertices on $P$.
\end{itemize}
This completes the argument.

The simplification for the shortest path problem follows analogously. This section demonstrates that for future work on parallel reachability or shortest paths, one can directly apply our black-box reduction and focus only on shallow graphs.
\end{document}